\documentclass{article}

\usepackage{arxiv}

\usepackage[utf8]{inputenc} 
\usepackage[T1]{fontenc}    
\usepackage{hyperref}       
\usepackage{url}            
\usepackage{booktabs}       
\usepackage{amsfonts}       
\usepackage{nicefrac}       
\usepackage{microtype}      
\usepackage{lipsum}

\usepackage{cite}
\usepackage{amsmath,amssymb,amsfonts}
\usepackage{graphicx}
\usepackage{algorithm}
\usepackage[noend]{algpseudocode}
\usepackage{hyperref}
\usepackage{textcomp}
\usepackage{threeparttablex} 
\usepackage{adjustbox}
\usepackage{mathtools}
\usepackage{caption}
\usepackage{float}
\usepackage{stmaryrd}
\usepackage{epstopdf}
\usepackage{multirow}
\usepackage{pifont}
\usepackage{caption}
\usepackage{subcaption}
\usepackage{xcolor}

\newcommand{\norm}[1]{\left\lVert#1\right\rVert}

\newcommand{\R}{{\mathbb{R}}}
\newcommand{\Y}{{\textbf{Y}}}
\newcommand{\y}{{\textbf{y}}}

\newcommand{\X}{{\mathbf{X}}}
\newcommand{\W}{{\textbf{V}}}
\newcommand{\vo}{{\mathbf{v}}}
\newcommand{\T}{{\mathbf{T}}}
\newcommand{\So}{{\mathbf{S}}}

\newcommand{\con}{{\mathcal{o}}}
\newcommand{\U}{{\mathbf{U}}}

\newcommand{\cmark}{\ding{51}}%
\newcommand{\xmark}{\ding{55}}%

\newtheorem{theorem}{Theorem}[section]

\newtheorem{assumption}{Assumption}

\newtheorem{definition}[theorem]{Definition}
\newtheorem{lemma}[theorem]{Lemma}
\newtheorem{remark}[theorem]{Remark}
\newtheorem{problem}[theorem]{Problem}
\newtheorem{proof}[theorem]{Proof}

\title{Spatiotemporal Tubes based Control of Unknown Multi-Agent Systems for Temporal Reach-Avoid-Stay Tasks
\thanks{$^\ddag$Authors contributed equally.}
\thanks{ A. Basu, R. Das, and P. Jagtap  are with the Robert Bosch Centre for Cyber-Physical Systems, IISc, Bangalore, India {\tt\footnotesize \{ahanbasu,ratnangshud,pushpak\}@iisc.ac.in}}
}

\author{
 Ahan Basu$^{\ddag}$ \\
  Robert Bosch Centre for Cyber-Physical Systems\\
  IISc, Bengaluru, India\\
  \texttt{siddharthau@iisc.ac.in} \\
   \And
 Ratnangshu Das$^{\ddag}$ \\
  Robert Bosch Centre for Cyber-Physical Systems\\
  IISc, Bengaluru, India\\
  \texttt{ratnangshud@iisc.ac.in} \\
   \And
 Pushpak Jagtap \\
  Robert Bosch Centre for Cyber-Physical Systems\\
  IISc, Bengaluru, India\\
  \texttt{pushpak@iisc.ac.in} \\
}

\begin{document}
\maketitle

\begin{abstract}
The paper focuses on designing a controller for unknown dynamical multi-agent systems to achieve temporal reach-avoid-stay tasks for each agent while preventing inter-agent collisions. The main objective is to generate a spatiotemporal tube (STT) for each agent and thereby devise a closed-form, approximation-free, and decentralized control strategy that ensures the system trajectory reaches the target within a specific time while avoiding time-varying unsafe sets and collisions with other agents. In order to achieve this, the requirements of STTs are formulated as a robust optimization problem (ROP) and solved using a sampling-based scenario optimization problem (SOP) to address the issue of infeasibility caused by the infinite number of constraints in ROP. The STTs are generated by solving the SOP, and the corresponding closed-form control is designed to fulfill the specified task. Finally, the effectiveness of our approach is demonstrated through two case studies, one involving omnidirectional robots and the other involving multiple drones modelled as Euler-Lagrange systems.
\end{abstract}

\section{Introduction}
Over the past few decades, control of multi-agent systems (MAS) has gained significant attention in the control community. These systems have a wide range of applications, from satellite formation flying \cite{ahn2012} to power grids \cite{mokhtari2013}, making coordination crucial to handle collision situations. Distributed control has emerged as one of the most effective strategies for multi-agent systems, with successful applications in formation control \cite{chen2019control} and tracking tasks \cite{Distributed}. However, these techniques are generally unsuitable for reach-avoid-stay (RAS) specifications with temporal constraints.

The RAS specifications \cite{Meng1} have become increasingly important for autonomous systems, particularly in ensuring that these systems reach a desired target set from a specific initial set while avoiding unsafe areas and adhering to state constraints. These specifications are critical for applications in safety-critical systems such as self-driving cars and unmanned aerial vehicles. They also serve as foundational elements for more complex requirements \cite{Kloetzer,NAHS,jagtapCBC}, making the development of safe and reliable control strategies for these specifications essential. In this work, in addition to RAS tasks, we also take into account additional temporal constraints where the system trajectories should reach the target region within a specific time, and unsafe regions are time-dependent. We consider that each agent in MAS has been assigned such temporal RAS tasks independently, in addition to ensuring collision avoidance between agents. Handling such temporal RAS tasks in MAS settings is particularly challenging due to the need for scalable, decentralized control laws that ensure agent-level specification satisfaction while enforcing inter-agent safety.

To synthesize the controller for the RAS tasks, symbolic control techniques \cite{tabuda} have become an important tool due to their ability to provide formal guarantees. These approaches have been implemented for satisfying task specifications in complex multi-agent systems \cite{Symbolic_MA}; however, the curse of dimensionality remains a significant challenge. Some studies have integrated control barrier functions (CBFs) \cite{CBF,jagtapCBC} within a symbolic control framework to address scalability issues in MAS \cite{CBF_Symbolic_Multi}. Despite these efforts, the computational cost still increases exponentially with the number of agents and the complexity of the overall system. Additionally, these techniques struggle to handle temporal constraints, such as ensuring prescribed-time reachability or avoidance of time-varying unsafe sets, which are critical for safety in real-world autonomous systems.
Recently, quadratic programming-based CBF, a discretization-free approach, has been used to deal with a class of temporal constraints in multi-agent systems \cite{lindemann2020barrier}. However, these techniques often rely on exact knowledge of system dynamics, which can be challenging in practical scenarios.

Funnel-based control \cite{PPC1} offers an approximation-free closed-form control law that ensures the trajectories remain within predefined funnel constraints. This technique has been adapted for multi-agent systems to handle a class of spatio-temporal logic tasks \cite{lindemann2019feedback, Funnel_STL_MAS}. Although it effectively handles specifications defined using convex predicates, such as reachability and tracking, it does not generalize well to specifications with nonconvex predicates, such as avoiding unsafe sets \cite{Funnel_STL, funnelarxiv}. To address prescribed-time reach while avoid specifications, the spatiotemporal tube (STT) approach \cite{STT, das2025control, das2025spatiotemporal, faruqui2025reach}, defined using continuously differentiable time-varying tube functions, is proposed for a single agent control-affine dynamics. 

Considering the aforementioned research gaps, this paper aims to design a decentralized, approximation-free, and closed-form control law for unknown multi-agent systems with bounded disturbance to achieve temporal reach-avoid-stay (T-RAS) tasks by each agent while preventing inter-agent collisions. To solve these problems, we introduce a sampling-based approach, inspired by \cite{MCBC, FV_DD}, to generate the STTs that meet T-RAS specifications for each agent. To ensure collision avoidance, we introduce an additional constraint to ensure that the STTs of each agent do not overlap. Our approach begins by formulating the conditions of STTs as a robust optimization problem (ROP). We then sample points in time and time-dependent unsafe regions to develop a scenario optimization problem (SOP) associated with the ROP. Solving the SOP enables the construction of STTs that capture both the T-RAS and collision avoidance specifications, with a correctness guarantee of 1. Following this, we design a closed-form, approximation-free, and decentralized control law to constrain the state trajectories of each agent within the obtained STTs, ensuring the achievement of both objectives. This framework provides a scalable solution to high-dimensional multi-agent problems under uncertainty and addresses limitations in prior work by simultaneously ensuring temporal task completion and safety. The effectiveness of our approach is demonstrated through various case studies.

The paper begins with the definition of multi-agent systems and problem formulation in Section \ref{preliminaries}. Section \ref{data_driven} introduces a sampling-based approach for constructing STTs by solving scenario optimization problems. The following section uses the obtained STTs to design closed-form and approximation-free controllers for each system, ensuring T-RAS task completion while preventing inter-agent collisions in a decentralized manner. Section \ref{case-study} validates the approach through two case studies and discusses potential future directions along with the conclusion in Section \ref{conclusion}.

\section{Preliminaries and Problem Formulation}
\label{preliminaries}

\begin{table}[ht]
\centering
\caption{Summary of Notation}
\begin{tabular}{ll}
\hline
\textbf{Symbol} & \textbf{Description} \\
\hline
$\mathbb{N}, \mathbb{R}, \mathbb{R}^+, \mathbb{R}_0^+$ & 
\begin{tabular}[t]{@{}l@{}}
Sets of natural, real, positive real, and \\
nonnegative real numbers
\end{tabular} \\
$\mathbb{R}^{n \times m}$ & Set of real matrices of size $n \times m$ \\
$\| \cdot \|$ & Euclidean norm \\
$[a; b]$ & Closed interval in $\mathbb{N}$, $a \leq b$ \\
$[x_1, \ldots, x_n]^\top$ & Column vector in $\mathbb{R}^n$ \\
$\textsf{diag}(d_1,\ldots,d_n)$ & Diagonal matrix with specified entries \\
$M^\top$ & Transpose of matrix $M$ \\
$\mathcal{P}(\mathbf{A})$ & Power set of set $\mathbf{A}$ \\
$\prod_{i=1}^{N} \X_i$ & Cartesian product of sets $\X_i$ \\
$[\X_{i,L}, \X_{i,U}]$ & 
\begin{tabular}[t]{@{}l@{}}
Projection interval of set $\X$ onto $i$-th coordinate, \\
$\X_{i,L} := \min\{x_i \in \mathbb{R} \mid [x_1, \ldots, x_n] \in \X\}$, \\
$\X_{i,U} := \max\{x_i \in \mathbb{R} \mid [x_1, \ldots, x_n] \in \X\}$
\end{tabular} \\
$x \circ y$ & Elementwise multiplication of vectors $x,y \in \mathbb{R}^n$, \\
& defined by $(x \circ y)_i := x_i y_i, \forall i \in [1; n]$ \\
$I_n$ & Identity matrix of size $n \times n$ \\
\hline
\end{tabular}
\end{table}

\subsection{System Definition}
Consider multi-agent systems with $M$ agents, whose $j$-th agent is given by a class of control-affine MIMO nonlinear pure feedback systems defined as:
\begin{align} \label{eqn:sysdyn}
    &\dot{x}_i^{(j)}(t) = f_i^{(j)}(\overline{x}_i^{(j)}(t)) + g_i^{(j)}(\overline{x}_i^{(j)}(t))x_{i+1}^{(j)}(t) + w_i^{(j)}(t),\notag \\ &i\in [1;N-1], \notag\\
    &\dot{x}_{N}^{(j)}(t) = f_N^{(j)}(\overline{x}_N^{(j)}(t)) + g_N^{(j)}(\overline{x}_N^{(j)}(t))u^{(j)}(t) + w_N^{(j)}(t), \\
    &y^{(j)}(t) = x_1^{(j)}(t), \nonumber
\end{align}
where for $t\in\R^+_0$, $i\in[1;N]$, and $j \in [1;M]$,
\begin{itemize}
    \item $x_i^{(j)}(t) = [x_{i,1}^{(j)}(t), \ldots, x_{i,n}^{(j)}(t)]^\top \in {\X}_i^{(j)} \subset \mathbb{R}^{n}$ is the state vector,
    \item $\overline{x}_i^{(j)}(t) := [{x_1^{(j)}}^\top(t),\ldots,{x_i^{(j)}}^\top(t)]^\top \in \overline{\X}_i^{(j)} = \prod_{k=1}^i \X_i^{(j)} \subset \mathbb{R}^{ni} $,
    \item $u^{(j)}(t) \in \mathbb{R}^n$ is the control input vector,
    \item $w_i^{(j)}(t) \in \mathbf{W} \subset \R^n$ is an unknown bounded external disturbance, and
    \item $y^{(j)}(t) = [x_{1,1}^{(j)}(t), \ldots, x_{1,n}^{(j)}(t)]\in \Y^{(j)}=\X_1^{(j)}$ denotes the output vector.
\end{itemize}

For all $i \in [1;N]$ and $j \in [1;M]$, the functions $f_i^{(j)}: \mathbb{R}^{n} \rightarrow \mathbb{R}^n$ and $ g_i^{(j)}: \mathbb{R}^{n} \rightarrow \mathbb{R}^{n \times n}$ follows the Assumptions \ref{assum:lip} and \ref{assum:pd}.

\begin{assumption}\label{assum:lip}
    For all $i \in [1;N], j \in [1;M]$, functions $f_i^{(j)}$ and $g_i^{(j)}$ are unknown and locally Lipschitz.         
\end{assumption}
\begin{assumption}[\cite{PPC1}, \cite{STT}, \cite{RAC}]\label{assum:pd}
    For all $j \in [1;M]$, the matrix $g_{i,s}^{(j)}(\overline{x}_i^{(j)}) = \frac{g_i^{(j)}(\overline{x}_i^{(j)})+g_i^{(j)}(\overline{x}_i^{(j)})^\top}{2}$ is uniformly sign definite with known signs for all $\overline{x}_i^{(j)} \in \overline{\X}_i^{(j)}$. Without loss of generality, we assume $g_{i,s}^{(j)}(\overline{x}_i^{(j)})$ is positive definite, i.e., there exists a constant $\underline{g_i^{(j)}}\in\mathbb R^+, \forall i \in [1;N]$ such that
    $$0 < \underline{g_i^{(j)}} \leq \lambda_{\min} (g_{i,s}^{(j)}(\overline{x}_i^{(j)})), \forall \ \overline{x}_i^{(j)} \in \overline{\X}_i^{(j)},$$
    where $\lambda_{\min}(\cdot)$ represents the smallest eigenvalue of the matrix.
    This assumption establishes a global controllability condition for \eqref{eqn:sysdyn}.
\end{assumption}

\subsection{Problem Formulation}
Let the output of the $j$-th system \eqref{eqn:sysdyn}, denoted by $y^{(j)}(t)$, be subject to a \textit{temporal reach-avoid-stay specification} defined below.

\begin{definition}[Temporal reach-avoid-stay (T-RAS) task]\label{def:tras}
Given a common output space $\Y = \X_{1}^{(j)}, \forall j \in [1;M]$, a prescribed time horizon $[0, t_c]$ with $t_c \in \R^+$, a time-varying unsafe set $\U: \R_0^+ \rightarrow \mathcal{P}(\Y)$, an initial set $\So^{(j)} \subset \Y \setminus \U(0)$, and a target set $\T^{(j)} \subset \Y \setminus \U(t_c)$, we say that the output $y^{(j)}(t)$ of agent $j$ satisfies the T-RAS task if:
\begin{itemize}
    \item $y^{(j)}(0) \in \So^{(j)}$ (starts in the safe initial set),
    \item $y^{(j)}(t_c) \in \T^{(j)}$ (reaches the target set at time $t_c$),
    \item $y^{(j)}(t) \in \Y \setminus \U(t), \ \forall t \in [0, t_c]$ (remains outside the time-varying unsafe set).
\end{itemize}
\end{definition}
\begin{remark}
    The T-RAS task ensures avoidance of both static and dynamic obstacles. These are captured by the time-varying set $\U(t) \subset \Y$, which aggregates all obstacle regions.
\end{remark}
\begin{remark}
    Note that the output space $\Y$ is the same for all the agents, and the unsafe set is present only in the output space.
\end{remark}
\begin{definition}[Inter-agent Collision Avoidance (CA) task] \label{def:ca}
Given $M$ agents with outputs $y^{(j)}(t)$ for $j \in [1;M]$, we say the agents satisfy the CA task if:
$\forall t \in [0, t_c], \quad \forall j \neq k, \quad \|y^{(j)}(t)- y^{(k)}(t)\| > 0.$
\end{definition}
\begin{remark}
    The CA task specifically refers to inter-agent collision avoidance. It ensures that agents maintain safe separation from each other throughout the task horizon. The avoidance of static and dynamic obstacles in the environment is separately handled by the T-RAS task via the set $\U(t)$.
\end{remark}

\begin{problem}\label{prob}
Given $M$ agents governed by the dynamics in \eqref{eqn:sysdyn}, design an approximation-free, closed-form, and decentralized control law such that each agent's output $y^{(j)}(t)$ satisfies:
\begin{itemize}
    \item The T-RAS specification in Definition \ref{def:tras}, and
    \item The CA specification in Definition \ref{def:ca}.
\end{itemize}
\end{problem}


In order to solve this problem, we leverage the notion of spatiotemporal tubes, proposed by authors in \cite{STT}, to address the T-RAS specification for a single control affine agent and defined as follows:

\begin{definition}[Spatiotemporal tubes (STT) for T-RAS task]\label{def:stt}
   Given a T-RAS task in Definition \ref{def:tras}, for $j$-th agent, time-varying intervals $[\gamma_{i,L}^{(j)}(t), \gamma_{i,U}^{(j)}(t)]$, where $\gamma_{i,L}^{(j)}:\R_0^+\rightarrow\R$ and $\gamma_{i,U}^{(j)}:\R_0^+\rightarrow\R$ are continuously differentiable functions with $\gamma_{i,L}^{(j)}(t) < \gamma_{i,U}^{(j)}(t)$, are called spatiotemporal tubes (STT) for T-RAS task if for all $i \in [1;n]$, the following holds:
\begin{subequations}\label{eqn:stt}
\begin{align}
    &\prod_{i=1}^n [\gamma_{i,L}^{(j)}(0), \gamma_{i,U}^{(j)}(0)] \subseteq \So^{(j)},
    \prod_{i=1}^n [\gamma_{i,L}^{(j)}(t_c), \gamma_{i,U}^{(j)}(t_c)] \subseteq \T^{(j)}, \label{eqn:sttbc} \\
    &\prod_{i=1}^n [\gamma_{i,L}^{(j)}(t), \gamma_{i,U}^{(j)}(t)] \subseteq \Y, \forall t \in [0,t_c], \label{eqn:stta}\\
    &\prod_{i=1}^n 
    [\gamma_{i,L}^{(j)}(t), \gamma_{i,U}^{(j)}(t)] \cap \U(t) = \emptyset, \forall t \in [0,t_c].\label{eqn:sttd}
\end{align}
\end{subequations}
\end{definition}

So, using the STTs, we aim to solve the T-RAS specification for each agent while performing the global CA task. 

\begin{remark}
    If one designs a control law that constrains the output trajectory within the STTs defined in Definition \ref{def:stt}, i.e.,
    \begin{align} \label{eqn:stt_con}
    &\gamma_{i,L}^{(j)}(t) < y_i^{(j)}(t) < \gamma_{i,U}^{(j)}(t), \forall i \in [1;n], \forall j \in [1;M] \nonumber \\
    &\implies \gamma_L^{(j)}(t) < x_1^{(j)}(t) < \gamma_U^{(j)}(t),
    \end{align}
then one can ensure the satisfaction of T-RAS specification.
\end{remark}

The tubes in \eqref{eqn:stt} can not accomplish the global collision avoidance (CA) task for multi-agent systems \eqref{eqn:sysdyn}, as they do not inherit any constraint to prevent the collision task. Hence, the safety of the agents is at stake. So, to incorporate the collision prevention task, we introduce the following constraint in addition to \eqref{eqn:stt}. 
\begin{align} \label{eq:collision}
    &\forall j, k \in [1;M] \text{ and } j \neq k, \forall t \in [0,t_c] : \notag \\
    &\prod_{i=1}^n[\gamma_{i,L}^{(j)}(t), \gamma_{i,U}^{(j)}(t)] \bigcap \prod_{i=1}^n[\gamma_{i,L}^{(k)}(t), \gamma_{i,U}^{(k)}(t)] = \emptyset. 
\end{align}
Thus, Equation \eqref{eqn:stt} and Equation \eqref{eq:collision} can simultaneously handle the T-RAS and CA task defined in Definition \ref{def:tras} and \ref{def:ca}.

The STT approach proposed in \cite{STT} adds to the conservatism in the specifications by limiting the unsafe set to a union of convex sets and requiring the projection of the unsafe set not to overlap with the initial and target sets in at least one dimension. Moreover, the introduction of the circumvent function results in an abrupt change in the tube shape, thereby increasing the control effort. Therefore, we introduce the sampling-based approach to construct the STTs for the multi-agent systems.

\section{Construction of Spatiotemporal Tubes} \label{data_driven}
In this section, our main focus is to synthesize the STTs for multi-agent systems, \textit{i.e.,} to construct the tubes for each agent that will reach the target region, avoiding the unsafe sets as well as avoiding collision with each other. We use the sampled data of the unsafe set and time to generate the tubes.

\begin{lemma}\label{lem:CA_condition}
    Equation \eqref{eq:collision} implies that for all $ j,k \in [1;M], \hspace{0.1 cm} j \neq k$, there exists $i\in[1;n]$, such that:
    \begin{equation} \label{eq:col}
        \text{min} \left (\gamma_{i,U}^{(j)}(t) - \gamma_{i,L}^{(k)}(t), \gamma_{i,U}^{(k)}(t) - \gamma_{i,L}^{(j)}(t) \right) < 0.
    \end{equation}
\end{lemma}
\begin{proof}
   Let $y(t) = [y_1(t), \ldots y_n(t)]$ be any particular point inside the interval $\prod_{i=1}^n[\gamma_{i,L}^{(j)}(t), \gamma_{i,U}^{(j)}(t)]$ at time $t \in [0,t_c]$. From Equation \eqref{eq:collision}, it can be infered that $y(t) \not \in \prod_{i=1}^n[\gamma_{i,L}^{(k)}(t), \gamma_{i,U}^{(k)}(t)], \text{for any} \hspace{0.1 cm} j \neq k \hspace{0.1 cm}\text{at time}\hspace{0.1 cm} t$.

   Clearly, there exists $i$ such that $x_{1,i}(t) \in [\gamma_{i,L}^{(j)}(t), \gamma_{i,U}^{(j)}(t)]$, but $x_{1,i}(t) \not \in [\gamma_{i,L}^{(k)}(t),\gamma_{i,U}^{(k)}(t)]$. From here, one can imply that either $\gamma_{i,U}^{(j)}(t) < \gamma_{i,L}^{(k)}(t)$ or $\gamma_{i,U}^{(k)}(t) < \gamma_{i,L}^{(j)}(t)$ that is equivalent to \eqref{eq:col}, which completes the proof.
\end{proof}

In this setting, we fix the structure of the curves that will form the STTs for $j$-th agent and for the $i$-th dimension as, 
$$\gamma_{i,\con}^{(j)}(c_{i,\con}^{(j)},t) = \sum_{k=1}^{z_{i,\con}} c_{i,\con}^{j,k} p_{i,\con}^k(t), \ \con \in \{L,U\}, \ i\in [1;n], \ j \in [1;N].$$
where $L$ and $U$ denote the lower and upper constraints, respectively. $p_{i,\con}(t)$ are user-defined locally Lipschitz continuous basis functions and $c_{i,\con}^{(j)} = [c_{i,\con}^{j,1}; c_{i,\con}^{j,2};\ldots; c_{i,\con}^{j,z_{i,\con}}] \in \mathbb{R}^{z_{i,\con}}$ denote the unknown coefficients.

To satisfy the conditions in Definition \ref{def:stt} and \eqref{eq:collision}, we recast our problem as the following robust optimization problem (ROP):
\begin{subequations} \label{eq:ROP}
\begin{align}
& \min_{[e_1, e_2,\ldots,e_N,\eta]} \quad \eta  \notag \\
& \textrm{s.t.} \notag \\
& \forall i \in[1;n], \forall j \in[1;M], \notag \\
& \quad \gamma_{i,\con}^{(j)}(c_{i,\con}^{(j)},0)\hspace{-0.2em} =\hspace{-0.2em} \hat{\So}_{i,\con}^{(j)}, \gamma_{i,\con}^{(j)}(c_{i,\con}^{(j)},t_c) = \hat{\T}_{i,\con}^{(j)}; \label{eq:start_target_ROP}
\end{align}
\begin{align}
& \forall (t,i) \in [0,t_c]\times[1;n], \forall j \in [1;M]: \notag \\
& \quad \Y_{i,L} - \gamma_{i,\con}^{(j)}(c_{i,\con}^{(j)},t)  \leq \eta_i^{(j)}, 
\gamma_{i,\con}^{(j)}(c_{i,\con}^{(j)},t) - \Y_{i,U}  \leq \eta_i^{(j)}, \label{eq:state_boundary_ROP} \\
& \quad \gamma_{i,L}^{(j)}(c_{i,L}^{(j)},t) - \gamma_{i,U}^{(j)}(c_{i,U}^{(j)},t) + \gamma_{i,d}^{(j)} \leq \eta_i^{(j)},\label{eq:curve_boundary_ROP} \\
& \forall (t, \y) \in [0,t_c] \times \U(t), \forall j \in [1;M] \hspace{0.1 cm}, \exists \hspace{0.1 cm} i \in [1;n]: \notag \\
& \quad \text{min} \left (\y_i - \gamma_{i,L}^{(j)}(c_{i,L}^{(j)},t), \gamma_{i,U}^{(j)}(c_{i,U}^{(j)},t) - \y_i   \right) \leq \eta_i^{(j)}; \label{eq:unsafe_ROP}\\
& \forall t \in [0,t_c], \forall j,k \in [1;M] \hspace{0.1 cm} \text{and}\hspace{0.1 cm} j \neq k, \hspace{0.1 cm} \exists \hspace{0.1 cm} i \in [1;n]: \notag \\
& \quad \text{min} \big\{\gamma_{i,U}^{(j)}(c_{i,U}^{(j)},t) - \gamma_{i,L}^{(k)}(c_{i,L}^{(k)},t), \notag \\
& \quad \quad \gamma_{i,U}^{(k)}(c_{i,U}^{(k)},t) - \gamma_{i,L}^{(j)}(c_{i,L}^{(j)},t) \big\}\leq \eta_i^{(j)}; \label{eq:CA_ROP}\\
& \forall i \in[1;n], \forall j \in[1;M]: \eta_i^{(j)} < \eta, \label{eq:eta_ROP}\\
& d_i^{(j)} = [c_{i,L}^{(j)}, c_{i,U}^{(j)}, \eta_i^{(j)}], \
e_j = [d_1^{(j)}, \ldots, d_n^{(j)}] \notag .
\end{align}
\end{subequations}

Here, $\con \in \{L,U\}$. One can readily observe that if the solution to the ROP $\eta^* \leq 0$, then it ensures the conditions in Definition \ref{def:tras} and \ref{def:ca} will be satisfied along with \eqref{eq:collision}.

\begin{remark}\label{rem:cons_explain}
    {If the solution to ROP $\eta^* \leq 0$, then the conditions~\eqref{eq:ROP} directly correspond to the STT characteristic in Equation~\eqref{eqn:stt} and \eqref{eq:collision}.}
    \begin{enumerate}
        \item {The condition \eqref{eq:start_target_ROP} guarantees \eqref{eqn:sttbc}, i.e., the tubes of $j$-th agent start within the given initial set $\So^{(j)}$:
        $$\prod_{i=1}^n[\gamma_{i,L}^{(j)}(c_{i,L}^{(j)},0), \gamma_{i,U}^{(j)}(c_{i,U}^{(j)},0)] = \prod_{i=1}^n [\hat{\So}_{i,L}^{(j)}, \hat{\So}_{i,U}^{(j)}] \subseteq \So^{(j)},$$
        and reach their target set $\T^{(j)}$ at time $t = t_c$,
        $$\prod_{i=1}^n[\gamma_{i,L}^{(j)}(c_{i,L}^{(j)},t_c), \gamma_{i,U}^{(j)}(c_{i,U}^{(j)},t_c)] = \prod_{i=1}^n [\hat{\T}_{i,L}^{(j)}, \hat{\T}_{i,U}^{(j)}] \subseteq \T^{(j)}.$$
        \item Condition \eqref{eq:state_boundary_ROP} enforces that the tubes are always confined within the output-space $\Y$, thereby satisfying \eqref{eqn:stta}.
        \item Condition \eqref{eq:curve_boundary_ROP} with $\gamma_{i,d}^{(j)} \in \R^+$ ensures a non-zero separation between the STT boundaries, i.e., 
        $\gamma_{i,L}^{(j)}(t) < \gamma_{i,U}^{(j)}(t).$
        This provides a finite-width tube to contain the system trajectories.
        \item In condition \eqref{eq:unsafe_ROP}, requiring the minimum of the two terms to be less than a negative variable $\eta_i^{(j)}$ guarantees that the STT boundaries never intersect the unsafe set: the lower boundary remains above unsafe points or the upper boundary remains below them for all time. Consequently, the entire tube envelope stays outside the unsafe region, satisfying condition \eqref{eqn:sttd}.
        \item Finally, condition \eqref{eq:CA_ROP} ensures Lemma \ref{lem:CA_condition}, thereby satisfying the collision-avoidance requirement in \eqref{eq:collision}.}
    \end{enumerate}
\end{remark}

\begin{remark}\label{rem:eta_explain}
    {Here, $\eta_i^{(j)}$ denotes the local slack variable for the $i$-th dimension of agent $j$, and $\eta$ represents their global upper bound. If the ROP is solved with $\eta \leq 0$, then equation~\eqref{eq:eta_ROP} guarantees that all local slacks satisfy $\eta_i^{(j)} < 0$. As a result, every constraint is enforced, ensuring the construction of STTs that satisfy the T-RAS tasks for each agent while also fulfilling the CA requirement.}  
\end{remark}

\begin{remark}\label{rem:feas_ROP}
    {The proposed solution is sound: whenever a feasible solution exists, the resulting tubes guarantee that each agent satisfies the T-RAS tasks while maintaining CA. The feasibility of the optimization problem \eqref{eq:ROP}, however, depends on the degree of the polynomial basis functions used to parametrize the STTs. If no feasible solution is obtained, a higher-degree polynomial may be required to represent the tubes.}
\end{remark}

There are two major challenges in solving the proposed ROP in \eqref{eq:ROP}. Since the time is continuous, there will be infinitely many constraints. Another problem is that though the unsafe set is known a priori, constructing the tubes for each agent and avoiding all the points in the unsafe set will lead to solving an infinite number of constraints as well. 

To address this problem, we first define the augmented unsafe set as $\W = (t, \U) \in [0,t_c] \times \Y$. We then draw $N_t$ samples {$\vo_r = (t_r, \y_r)$ from $\W$, where $r \in [1;N_t]$, and construct open balls $\W_r $ centered at each sample $\vo_r$ with radius $\epsilon$. These balls are chosen such that for every point $(t,\y) \in \W$, there exists a sample $\vo_r$ satisfying
\begin{align}\label{eq:ball}
    \lVert (t, \y) - (t_r, \y_r) \rVert \leq \epsilon , \forall (t,\y) \in \W_r,
\end{align}
ensuring that the entire augmented unsafe set is covered by the union of these balls, i.e., 
$[0,t_c] \times \U(t) \subseteq \bigcup_{r=1}^{N_t} \W_r.$}

Now, we construct the scenario optimization problem (SOP) associated with the ROP in \eqref{eq:ROP} as follows:
\begin{subequations} \label{eq:SOP}
\begin{align}
& \min_{[e_1, e_2,\ldots,e_N,\eta]} \quad \eta  \notag \\
& \textrm{s.t.} \notag \\
& \forall i \in [1;n], \forall j \in[1;M]: \notag \\
& \quad \gamma_{i,\con}^{(j)}(c_{i,\con}^{(j)},0)\hspace{-0.2em} =\hspace{-0.2em} \hat{\So}_{i,\con}^{(j)}, \gamma_{i,\con}^{(j)}(c_{i,\con}^{(j)},t_c) = \hat{\T}_{i,\con}^{(j)}; \\
& \forall r \in [1;N_t], \forall (t_r,i) \in [0,t_c]\times[1;n], \forall j \in [1;M]: \notag \\
& \quad \Y_{i,L} - \gamma_{i,\con}^{(j)}(c_{i,\con}^{(j)},t_r)  \leq \eta_i^{(j)}, 
\gamma_{i,\con}^{(j)}(c_{i,\con}^{(j)},t_r) - \Y_{i,U}  \leq \eta_i^{(j)}, \\
& \quad \gamma_{i,L}^j(c_{i,L}^{(j)},t_r) - \gamma_{i,U}^{(j)}(c_{i,U}^{(j)},t_r) + \gamma_{i,d}^{(j)} \leq \eta_i^{(j)}, \\
& \forall r \in [1;N_t], \forall (t_r, \y_r) \in \W_r, \forall j \in [1;M] \hspace{0.1 cm}, \exists \hspace{0.1 cm} i \in [1;n]: \notag \\
& \quad \text{min} \left (x_{i,r} - \gamma_{i,L}^{(j)}(c_{i,L}^{(j)},t_r), \gamma_{i,U}^{(j)}(c_{i,U}^{(j)},t_r) - x_{i,r}   \right) \leq \eta_i^{(j)}; \\
& \forall r \in [1;N_t], \forall t_r \in [0,t_c], \forall j,k \in [1;M] \hspace{0.1 cm} \& \hspace{0.1 cm} j \neq k, \hspace{0.1 cm} \exists \hspace{0.1 cm} i \in [1;n]: \notag \\
& \quad \text{min} \big\{\gamma_{i,U}^{(j)}(c_{i,U}^{(j)},t_r) - \gamma_{i,L}^{(k)}(c_{i,L}^{(k)},t_r), \notag \\
& \quad \quad \gamma_{i,L}^{(k)}(c_{i,U}^{(k)},t_r) - \gamma_{i,L}^{(j)}(c_{i,L}^{(j)},t_r) \big\}\leq \eta_i^{(j)}; \\
& \forall i \in [1;n], \forall j \in[1;M]: \eta_i^{(j)} < \eta, \\
& d_i^{(j)} = [c_{i,L}^{(j)}, c_{i,U}^{(j)}, \eta_i^{(j)}], \notag \\
& e_j = [d_1^{(j)}, \ldots, d_n^j] \notag .
\end{align}
\end{subequations}

Here, $\y_r = [\y_{1,r}, \ldots, \y_{n,r}]$. {The SOP in \eqref{eq:SOP} therefore contains a finite number of constraints, each having the same structural form as those in the ROP in \eqref{eq:ROP}. In this way, the SOP serves as a finite approximation of the ROP, while preserving the feasibility conditions \eqref{eqn:stt} and \eqref{eq:collision} as discussed in Remark~\ref{rem:cons_explain}.}

Let the optimal solution of the SOP be $\eta_S^*$. Now, to guarantee that the tubes formed by solving the SOP \eqref{eq:SOP}, fulfills the constraints of ROP \eqref{eq:ROP}, we assume the following: 

\begin{assumption} \label{assum:funlip}
    $ \gamma_{i,L}^{(j)}(c_{i,L}^{(j)},t)$ and $ \gamma_{i,U}^{(j)}(c_{i,U}^{(j)},t)$ are Lipschitz continuous with respect to $t$ with Lipschitz constants $\mathcal{L}_{L}$ and $\mathcal{L}_{U}$ for any $i \in [1;n], j \in [1;M]$. Note that imposing these Lipschitz constraints prevents the sharp changes in the tube, resulting in less control effort.
\end{assumption}

Under Assumption \ref{assum:funlip}, Theorem \ref{th:constr} outlines the methodology for constructing STTs with a certified confidence of 1.

\begin{theorem} \label{th:constr}
    Under Assumption \ref{assum:funlip}, suppose the SOP in (\ref{eq:SOP}) is solved with $N_t$ sampled data as in (\ref{eq:ball}). Let the optimal value be $\eta_S^*$. If 
    \begin{equation} \label{eq:satisfy}
        \eta_S^* + \mathcal{L}\epsilon \leq 0,
    \end{equation}
    where $\mathcal{L}$ = max\{$\mathcal{L}_{L}, \mathcal{L}_{U}, \mathcal{L}_{L} + \mathcal{L}_{U}, \mathcal{L}_{L}+1, \mathcal{L}_{U}+1$\}, then the obtained STTs $[ \gamma_{i,L}^{(j)}(c_{i,L}^{(j)},t),  \gamma_{i,U}^{(j)}(c_{i,U}^{(j)},t)], \forall i \in [1;n], \forall j \in [1;M]$ from the SOP in \eqref{eq:SOP} ensures that the conditions corresponding to the ROP \eqref{eq:ROP} will be satisfied, and therefore by virtue of Remark \ref{rem:cons_explain} and \ref{rem:eta_explain}, the conditions in Definition \ref{def:stt} and Equation \eqref{eq:collision} are satisfied.
\end{theorem}

\begin{proof}\label{proof:constr}
    First, we demonstrate that under condition \eqref{eq:satisfy}, the $ \gamma_{i,L}^{(j)}(c_{i,L}^{(j)},t)$ and $ \gamma_{i,U}^{(j)}(c_{i,U}^{(j)},t)$ constructed through solving the SOP in \eqref{eq:SOP} satisfy Equation \eqref{eqn:stt}(a) as well as Equation \eqref{eq:collision}. The optimal $\eta_S^*$, obtained through solving the \eqref{eq:SOP}, guarantees for any $t_r \in [0,t_c],$ we have: 
    \begin{align*}
    &\Y_{i,L} - \gamma_{i,L}^{(j)}(c_{i,L}^{(j)},t_r) \leq \eta_S^*, \\
    & \gamma_{i,U}^{(j)}(c_{i,U}^{(j)},t_r) - \Y_{i,U} \leq \eta_S^*,  \\
    &\gamma_{i,L}^{(j)}(c_{i,L}^{(j)}, t_r) -  \gamma_{i,U}^{(j)}(c_{i,U}^{(j)},t_r) +\gamma_{i,d}^{(j)} \leq \eta_S^*, \\
    &\text{min} \big( \gamma_{i,U}^{(j)}(c_{i,U}^{(j)},t_r) - \gamma_{i,L}^{(k)}(c_{i,L}^{(k)},t_r), \gamma_{i,U}^{(k)}(c_{i,U}^{(k)},t_r) - \\ 
    &\quad \gamma_{i,L}^{(j)}(c_{i,L}^{(j)},t_r) \big) \leq  \eta_S^* . 
    \end{align*}
    Now from \eqref{eq:ball} we infer that $\forall t \in [0,t_c], \exists \hspace{0.2em} t_r $ s.t. $|t-t_r| \leq \epsilon$. 
    Thus, $\forall i \in [1;n], \forall r \in [1;N_t], \forall j \in [1;M]$:
    \begin{itemize}
       \item[(a)] $\Y_{i,L} - \gamma_{i,L}^{(j)}(c_{i,L}^{(j)},t) =\Y_{i,L} - \gamma_{i,L}^{(j)}(c_{i,L}^{(j)},t_r) + \gamma_{i,L}^{(j)}(c_{i,L}^{(j)},t_r) - \gamma_{i,L}^{(j)}(c_{i,L}^{(j)},t) \leq \eta_S^* + \mathcal{L}_{L} |t-t_r| \leq \mathcal{L}\epsilon + \eta_S^* \leq 0$, \\
       \item[(b)] $ \gamma_{i,U}^{(j)}(c_{i,U}^{(j)},t) - \Y_{i,U} = \gamma_{i,U}^{(j)}(c_{i,U}^{(j)},t) -  \gamma_{i,U}^{(j)}(c_{i,U}^{(j)},t_r) + \gamma_{i,U}^{(j)}(c_{i,U}^{(j)},t_r) - \Y_{i,U} \leq \mathcal{L}_{U} |t-t_r| + \eta_S^* \leq \mathcal{L}\epsilon + \eta_S^* \leq 0$, \\
       \item[(c)] $ \gamma_{i,L}^{(j)}(c_{i,L}^{(j)},t) -  \gamma_{i,U}^{(j)}(c_{i,U}^{(j)},t) + \gamma_{i,d}^{(j)} = \big( \gamma_{i,L}^{(j)}(c_{i,L}^{(j)},t) - \gamma_{i,L}^{(j)}(c_{i,L}^{(j)},t_r)\big) + \big(\gamma_{i,L}^{(j)}(c_{i,L}^{(j)},t_r) - \gamma_{i,U}^{(j)}(c_{i,U}^{(j)},t_r) + \gamma_{i,d}^{(j)}\big) + \big( \gamma_{i,U}^{(j)}(c_{i,U}^{(j)},t_r) -  \gamma_{i,U}^{(j)}(c_{i,U}^{(j)},t)\big) \leq \mathcal{L}_{L} |t-t_r| +\eta_S^* + \mathcal{L}_{U} |t-t_r| \leq (\mathcal{L}_{L} + \mathcal{L}_{U})\epsilon + \eta_S^* \leq \mathcal{L}\epsilon + \eta_S^* \leq 0 $,
       \item[(d)] $ \gamma_{i,U}^{(j)}(c_{i,U}^{(j)},t) - \gamma_{i,L}^{(k)}(c_{i,L}^{(k)},t) = \big( \gamma_{i,U}^{(j)}(c_{i,U}^{(j)},t) -  \gamma_{i,U}^{(j)}(c_{i,U}^{(j)},t_r)\big) + \big( \gamma_{i,U}^{(j)}(c_{i,U}^{(j)},t_r) - \gamma_{i,L}^{(k)}(c_{i,L}^{(k)},t_r)\big) + \big(\gamma_{i,L}^{(k)}(c_{i,L}^{(k)},t_r) - \gamma_{i,L}^{(k)}(c_{i,L}^{(k)},t)\big) \leq \mathcal{L}_{U} |t-t_r| +\eta_S^* + \mathcal{L}_{L} |t-t_r| \leq (\mathcal{L}_{L} + \mathcal{L}_{U})\epsilon + \eta_S^* \leq \mathcal{L}\epsilon + \eta_S^* \leq 0$, \\
       \hspace{-2em}\text{Or} \\
       $\gamma_{i,U}^{(k)}(c_{i,U}^{(k)},t) - \gamma_{i,L}^{(j)}(c_{i,L}^{(j)},t) = \big (\gamma_{i,U}^{(k)}(c_{i,U}^{(k)},t) - \gamma_{i,U}^{(k)}(c_{i,U}^{(k)},t_r)\big) + \big (\gamma_{i,U}^{(k)}(c_{i,U}^{(k)},t_r) - \gamma_{i,L}^{(j)}(c_{i,L}^{(j)},t_r) \big) + \big (\gamma_{i,L}^{(j)}(c_{i,L}^{(j)},t_r) - \gamma_{i,L}^{(j)}(c_{i,L}^{(j)},t)\big) \leq \mathcal{L}_{U} |t-t_r| +\eta_S^* + \mathcal{L}_{L} |t-t_r| \leq (\mathcal{L}_{L} + \mathcal{L}_{U})\epsilon + \eta_S^* \leq \mathcal{L}\epsilon + \eta_S^* \leq 0$.
    \end{itemize} 
    
    Next, we show that when the condition in \eqref{eq:satisfy} is satisfied, the $ \gamma_{i,L}^{(j)}(c_{i,L}^{(j)},t)$ and $ \gamma_{i,U}^{(j)}(c_{i,U}^{(j)},t)$ formulated by solving the SOP in \eqref{eq:SOP} satisfy argument \eqref{eqn:stt}(d).
    The optimal $\eta_i^*$, obtained from solving the SOP in \eqref{eq:SOP}, also ensures that for any $(t_r, \y_r) \in \W$, 
    min $\left((\y_{i,r} - \gamma_{i,L}^{(j)}(c_{i,L}^{(j)},t_r), ( \gamma_{i,U}^{(j)}(c_{i,U}^{(j)},t_r) - \y_{i,r}\right) \leq \eta_i^*).$  
     Now, $\forall i \in [1;n],\forall j \in [1;M], \forall r \in [1;N_t]$:
     \begin{align*}
         & \quad \y_i - \gamma_{i,L}^{(j)}(c_{i,L}^{(j)},t) \\ 
         & = \left(\y_i - \y_{i,r}\right) + \left(\y_{i,r} - \gamma_{i,L}^{(j)}(c_{i,L}^{(j)},t_r)\right) + \big(\gamma_{i,L}^{(j)}(c_{i,L}^{(j)},t_r) \\
         &  - \gamma_{i,L}^{(j)}(c_{i,L}^{(j)},t)\big) \leq \epsilon + \eta_i^* +  \mathcal{L}_{L}|t-t_r| \\
         & \leq (\mathcal{L}_{L} + 1)\epsilon + \eta_S^* \leq \mathcal{L} \epsilon + \eta_S^* \leq 0 \\
         & \text{Or}\\
         & \quad  \gamma_{i,U}^{(j)}(c_{i,U}^{(j)},t) - \y_i \\
         & =  \left(\y_{i,r} - \y_i\right) + \left( \gamma_{i,U}^{(j)}(c_{i,U}^{(j)},t_r) - \y_{i,r}\right) + \big( \gamma_{i,U}^{(j)}(c_{i,U}^{(j)},t)  
         \end{align*}
         \begin{align*}
         &  -  \gamma_{i,U}^{(j)}(c_{i,U}^{(j)},t_r)\big) \leq \epsilon + \eta_i^* +  \mathcal{L}_{U}|t-t_r| \\
         & \leq (\mathcal{L}_{U} + 1)\epsilon + \eta_S^* \leq \mathcal{L} \epsilon + \eta_S^* \leq 0.
     \end{align*}

     Therefore, if the condition in \eqref{eq:satisfy} is met, the STT constructed with boundaries defined by $ \gamma_{i,L}^{(j)}(c_{i,L}^{(j)},t)$ and $ \gamma_{i,U}^{(j)}(c_{i,U}^{(j)},t), \forall i \in [1;n], \forall j \in [1;M]$ as determined by solving the SOP in \eqref{eq:SOP} satisfies Definition \ref{def:stt} and Equation \eqref{eq:collision}, thereby completing the proof.
\end{proof}

\begin{remark}
    Note that the Lipschitz constants $\mathcal{L}_L$ and $\mathcal{L}_U$ are required to check condition \eqref{eq:satisfy} in Theorem \ref{th:constr}. We introduce Algorithm \ref{algo:Lipschitz} to estimate these Lipschitz constants in the Appendix, which follows the similar procedure as \cite[Algorithm 1]{MCBC} and \cite[Algorithm 2]{FV_DD}.
\end{remark}

\section{Controller Design}
In this section, we use the STTs obtained from \eqref{eq:SOP} for each dimension of each agent to design an approximation-free, closed-form control law to solve Problem \ref{prob}. The lower triangular structure of \eqref{eqn:sysdyn} allows us to use a backstepping-like design approach similar to that described in \cite{feedback}. We show the design of the controller for the $j$-th agent where $j \in [1;M]$.

To simplify notation, we will omit the superscript $j$ and treat $\gamma_{i,L}^{(j)}(c_{i,L}^{(j)},t)$ as $\gamma_{i,L}(c_{i,L},t)$ and $\gamma_{i,U}^{(j)}(c_{i,U}^{(j)},t)$ as $\gamma_{i,U}(c_{i,U},t)$, for all $i = [1;n]$. Since the analysis in this section applies uniformly to all agents, this change will not impact readability.

 First, we design an intermediate control input $r_2$ for the $x_1$ dynamics to ensure the fulfillment of \eqref{eqn:stt_con}. We then iteratively design the intermediate control laws $r_{k+1}$ for the $x_k$ dynamics, ensuring $x_k$ tracks $r_k$, for all $k \in [2;N]$. It’s important to note that $r_{N+1}$ effectively becomes the actual control input $u$ for the system. The steps of the controller design are outlined below.

\textbf{Stage $1$}: {Given $\gamma_{1,i,L}(c_{i,L} ,t)$ and $\gamma_{1,i,U}(c_{i,U} ,t)$, $i\in[1;n]$,}
define the normalized error $e_1(x_1,t)$, the transformed error $\varepsilon_1(x_1,t)$ and the diagonal matrix $\xi_1(x_1,t)$ as
\begin{subequations} \label{eq:stage I}
   \begin{align}
    e_1(x_1,t) &= [e_{1,1}(x_{1,1},t), \ldots, e_{1,n}(x_{1,n},t)]^\top \notag \\
    &= (\gamma_{1,d} (t))^{-1} \left( 2x_1 - \gamma_{1,s} (t) \right), \\
    \varepsilon_1(x_1,t) &= \big[\ln\left(\frac{1+e_{1,1}(x_{1,1},t)}{1-e_{1,1}(x_{1,1},t)}\right), \ldots, \notag \\
    & \hspace{2em}\ln\left(\frac{1+e_{1,n}(x_{1,n},t)}{1-e_{1,n}(x_{1,n},t)}\right) \big]^\top ,\\
    \xi_1(x_1,t) &= 4 \gamma^{-1}_{1,d} (t) \Big(I_n- \text{diag} \big(e_1 \circ e_1 \big)\Big)^{-1},
    \end{align} 
\end{subequations}
where $\gamma_{1,s} := [\gamma_{1,1,U} + \gamma_{1,1,L}, \ldots, \gamma_{1,n,U} + \gamma_{1,n,L}]^\top$ and $\gamma_{1,d} := \textsf{diag}(\gamma_{1,1,d},\ldots,\gamma_{1,n,d})$, with $\gamma_{1,i,d} = \gamma_{1,i,U} - \gamma_{1,i,L}$.

The intermediate control input $r_2(x_1,t)$ is given by: 
\begin{equation*}
    r_2(x_1,t) = - \kappa_1\varepsilon_1(x_1,t)\xi_1(x_1,t), \kappa_1 \in \R^+.
\end{equation*}

\textbf{Stage $k$} ($k \in [2;N]$): {Given the reference vector $r_k(z_k,t)$, we aim to design the subsequent intermediate control $r_{k+1}(z_{k+1},t)$ for the dynamics of $x_{k+1}$, ensuring that $x_{k+1}$ tracks the trajectory determined by $r_k(z_k,t)$. Here, $z_k = [x_1,x_2,\ldots, x_{k-1}]^\top$.}

This is done by enforcing exponentially narrowing constraint $\gamma_{k,i}(t) = (p_{k,i} - q_{k,i})e^{-\mu_{k,i}t} + q_{k,i}$ such that 
\begin{align*}
    -\gamma_{k,i}(t) \leq (x_{k,i}-r_{k,i}) \leq \gamma_{k,i}(t) \ \ \forall i \in [1;n].
\end{align*}
Note that, $|x_{k,i}(t=0) - r_{k,i}(t=0)| \leq p_{k,i}$.

Now define the normalized error $e_k(x_{k},t)$, the transformed error $\varepsilon_k(x_{k},t)$ and the diagonal matrix $\xi_k(x_{k},t)$ as
\begin{subequations} \label{eq:stage k}
    \begin{align}
    e_k(x_{k},t) &= [e_{k,1}(x_{k,1},t), \ldots, e_{k,n}(x_{k,n},t)]^\top \notag \\
    &= (\gamma_{k,d} (t))^{-1} \left(x_{k} - r_k \right), \\
    \varepsilon_k(x_{k},t) &= \big[\ln\left(\frac{1+e_{k,1}(x_{k,1},t)}{1-e_{k,1}(x_{k,1},t)}\right), \ldots, \notag \\
    &\hspace{2em}\ln\left(\frac{1+e_{k,n}(x_{k,n},t)}{1-e_{k,n}(x_{k,n},t)}\right) \big]^\top, \\
    \xi_k(x_k,t) &= 4 \gamma^{-1}_{k,d} (t) \Big(I_n- \text{diag} \big(e_k \circ e_k \big)\Big)^{-1},
\end{align}
\end{subequations}
where $\gamma_{k,d} := \textsf{diag}(\gamma_{k,1,d},\ldots,\gamma_{k,n,d})$, with $\gamma_{k,i,d} = \frac{1}{2} (\gamma_{k,i,U} - \gamma_{k,i,L}) = \gamma_{k,i} \nonumber, \forall i \in [1;n]$.

So, the intermediate control inputs $r_{k+1}(z_{k+1},t)$ to enforce the desired temporal reach-avoid-stay task is given by 
\begin{equation*}
    r_{k+1}(z_{k+1},t) = - \kappa_k\varepsilon_k(x_{k},t)\xi_k(x_{k},t), \kappa_k \in \R^+.
\end{equation*}
Note that, at the $N$-th stage, $r_{N+1}(z_{N+1},t)$ essentially serves as the actual control input $u$, which is given by,
\begin{equation*}
    u(z_{N+1},t) = - \kappa_N\varepsilon_N(x_{N},t)\xi_N(x_{N},t), \kappa_N \in \R^+.
\end{equation*}





Thus, we can design the control input to ensure the trajectory for the system described in \eqref{eqn:sysdyn} to be bounded within the tubes derived using the SOP as discussed in Section \ref{data_driven}.

\begin{theorem} \label{theorem_inside}
    Given an agent described as nonlinear MIMO system in \eqref{eqn:sysdyn} satisfying assumptions \ref{assum:lip} and \ref{assum:pd} and spatio-temporal tubes as discussed in Section \ref{data_driven}, if the initial state of an agent $x_{k}(0)$ is within its corresponding STTs at time $t=0$, i.e., $\gamma_{k,i,L}(0) < x_{k,i}(0) < \gamma_{k,i,U}(0), \forall i \in [1;n], \forall k \in [1;N]$, then the closed-form control strategies,\begin{subequations}\label{eqn:Control_ras}
     \begin{align}
        r_{k+1}(z_{k+1},t) &= - \kappa_k\varepsilon_k(x_{k},t)\xi_k(x_{k},t), k \in [1;N-1], \\
        u(z_{N+1},t) &= - \kappa_N\varepsilon_N(x_{N},t)\xi_N(x_{N},t),
    \end{align}    
    \end{subequations}
    will ensure the trajectory of the agent bounded inside the corresponding tubes, where $z_{k+1} = [x_1,x_2,\ldots, x_k]^\top$ and $\varepsilon_k$, $\xi_k$ are as shown in \eqref{eq:stage I} and \eqref{eq:stage k}.
\end{theorem}

\begin{proof}
The proof is done for the stages mentioned above. First, we prove the control law for the STT, and then sequentially prove it for the other stages. 

In each stage, the proof proceeds in three steps. First, we show that there exists a maximal solution for the normalized error $e_k: [0,\tau_{\max}] \rightarrow \mathbb{D}$, which implies that $e_k(x_{k},t)$ remains within $\mathbb{D}$ in the maximal time solution interval $[0, \tau_{\max})$. Next, we show that the proposed control law in \eqref{eqn:Control_ras} constraints $e_k(x_{k},t)$ to a compact subset of $\mathbb{D}$. Finally, we prove that $\tau_{\max}$ can be extended to $\infty$.

\noindent\textbf{Stage $1$:}

Differentiating the normalized error $e_1(x_1,t)$ with respect to time and substituting the first equation of the system dynamics \eqref{eqn:sysdyn}, we get
\begin{align*}
     \dot{e}_1 =& 2(\gamma_{1,d}(t))^{-1} \Big( f_1(x_1) + g_1(x_1)x_2 + w_1  
     \\ & - \frac{1}{2}(\dot{\gamma}_{1,s}(t) + \dot{\gamma}_{1,d}(t) e_1)\Big) :=h_1(t,e_1),
\end{align*}
where $x_1 = \frac{\gamma_{1,d}(t)e_1 + \gamma_{1,s}(t)}{2}$.
We also define the constraints for $e_1$ through the open and bounded set $\mathbb{D}:=(-1, 1)^n$. 

\textit{\underline{Step $(i)$}}:
Since the initial state $x_1(0)$ satisfies $\gamma_{1,i,L}(0) < x_{1,i}(0) < \gamma_{1,i,U}(0), \forall i \in [1;n]$, the initial normalized error $e_1(0)$ is also within the constrained region $\mathbb{D}$. Further, the tube functions are bounded and continuously differentiable functions of time, the functions $f_1(x_1)$ and $g_1(x_1)$ are locally Lipschitz and the control law $r_2(x_1,t)$ is smooth over $\mathbb{D}$. As a consequence, $h_1(t,e_1)$ is bounded and continuously differentiable on $t$ and locally Lipschitz on $e_1$ over $\mathbb{D}$. 

Therefore, according to \cite[Theorem 54]{sontag}, there exists a maximal solution to the initial value problem $\dot{e}_1 = h_1(t,e_1), e_1(0) \in \mathbb{D}$ on the time interval $[0, \tau_{\max})$ such that $e_1(t) \in \mathbb{D}, \forall t \in [0, \tau_{\max}).$

\textit{\underline{Step $(ii)$}}:
Consider the following positive definite and radially unbounded Lyapunov function candidate: $V_1 = \frac{1}{2}\varepsilon_1^\top\varepsilon_1$. 

Differentiating $V_1$ with respect to $t$ and substituting $\dot{\varepsilon_1}$, $\dot{e_1}$, system dynamics \eqref{eqn:sysdyn}, and the control strategy (\ref{eqn:Control_ras}), we have:
\begin{align*}
    \dot{V}_1 &= \varepsilon_1^T \dot{\varepsilon}_1= \varepsilon_1^T \frac{2}{1-e_1^Te_1}\dot{e}_1 = \varepsilon_1^T \xi_1 \left(\dot{x}_1 - \frac{1}{2}(\dot{\gamma}_{1,s}+\dot{\gamma}_{1,d}e_1)\right) \\
    &= \varepsilon_1^T \xi_1 \left( f_1(x_1) + g_1(x_1)x_2 +w_1  - \frac{1}{2}(\dot{\gamma}_{1,s}+\dot{\gamma}_{1,d}e_1) \right), \\
    &= \varepsilon_1^T \xi_1 \left( f_1(x_1) -\kappa_1 g_1(x_1) \xi_1 \varepsilon_1 + w_1 -\frac{1}{2}(\dot{\gamma}_{1,s}+\dot{\gamma}_{1,d} e_1) \right).
\end{align*}
Using Rayleigh-Ritz inequality and Assumption \ref{assum:pd},
\begin{align*}
    \underline{g_1}\|\varepsilon_1\|^2 \|\xi_1\|^2 \leq \lambda_{\min}(g_1(x_1))\|\varepsilon_1\|^2 \|\xi_1\|^2  \leq \varepsilon_1^\top \xi_1 g_1(x_1) \xi_1 \varepsilon_1, \\
    -\kappa_1 \varepsilon_1^\top \xi_1 g_1(x_1) \xi_1 \varepsilon_1 \leq -\kappa_1 \underline{g_1}\|\varepsilon_1\|^2 \|\xi_1\|^2 = - \kappa_g^1\|\varepsilon_1\|^2 \|\xi_1\|^2.
\end{align*}
Therefore, 
$\dot{V}_1 \leq -\kappa_g^1\|\varepsilon_1\|^2 \|\xi_1\|^2 + \|\varepsilon_1\| \|\xi_1\| \|\Phi_1\|,$
where $\Phi_1 := f_1(x_1) + w_1 - \frac{1}{2}\dot{\gamma}_{1,s} - \frac{1}{2}\dot{\gamma}_{1,s}e_1$. 
From the construction of $\gamma_{1,s}, \gamma_{1,d}$  we know that $\dot{\gamma}_{1,s}$ and $\dot{\gamma}_{1,d}$ are bounded by construction. From step 1, we have $e_1(t) \in \mathbb{D}$ and, consequently, $\forall i \in [1;n]$, $x_{1,i}(t) \in (\gamma_{1,i,L}(t), \gamma_{1,i,U}(t))$. Thus, owing to the continuity of $f_1(x_1)$ and employing the extreme value theorem, we can infer $\|f_1(x_1)\| < \infty$. 
Hence, $\|\Phi_1\| < \infty, \forall t \in [0, \tau_{max})$. 

Now add and substract $\kappa_g^1\theta\norm{\varepsilon_1}^2 \|\xi_1\|^2$, where $\theta\in(0,1)$.
\begin{align*}
    &\dot{V}_1 \leq -\kappa_g^1(1-\theta)\norm{\varepsilon_1}^2 \|\xi_1\|^2 \\ & 
    \quad \quad \quad - \norm{\varepsilon_1}\|\xi_1\| \left(\kappa_g^1 \theta \norm{\varepsilon_1}\|\xi_1\| - \|\Phi_1\| \right) \\
    &\leq -\kappa_g^1(1-\theta)\norm{\varepsilon_1}^2 \norm{\xi_1}^2, \forall \kappa_g^1 \theta \norm{\varepsilon_1}\|\xi_1\| - \|\Phi_1\| \geq 0 \\
    &\leq -\kappa_g^1(1-\theta)\norm{\varepsilon_1}^2 \|\xi_1\|^2, \forall \norm{\varepsilon_1} \geq \frac{\|\Phi_1\|}{\kappa_g^1 \theta \|\xi_1\|} \nonumber,\forall t \in [0, \tau_{\max}).
\end{align*}
Therefore, we can conclude that there exists a time-independent upper bound $\varepsilon_1^* \in \mathbb{R}_{0}^+$ to the transformed error $\varepsilon_1$, i.e., $\|\varepsilon_1\| \leq \varepsilon_1^*, \forall t \in [0, \tau_{\max})$.

Consequently, taking inverse logarithmic function,
\begin{align*}
    -1 < \frac{e_{1,i}^{-\varepsilon_{1,i}^*}-1}{e_{1,i}^{-\varepsilon_{1,i}^*}+1} =: e_{1,i,L} \leq e_{1,i} \leq e_{1,i,U} := \frac{e_{1,i}^{\varepsilon_{1,i}^*}-1}{e_{1,i}^{\varepsilon_{1,i}^*}+1} < 1, \nonumber \\
    \forall t \in [0, \tau_{\max}), \ \text{for } i \in [1;n].
\end{align*}
Therefore, by employing the control law (\ref{eqn:Control_ras}), we can constrain $e_1$ to a compact subset of $\mathbb{D}$ as:
\begin{align} \label{eqn:e_compact}
    e_1(t) \in [e_{1,L}, e_{1,U}] =: \mathbb{D}' \subset \mathbb{D}, \forall t \in [0, \tau_{\max}), \nonumber
\end{align}
where $e_{1,L} = [e_{1,1,L}, \ldots, e_{1,n,L}]^\top$ and $e_{1,U} = [e_{1,1,U}, \ldots, e_{1,n,U}]^\top$.\\

\textit{\underline{Step $(iii)$}}:
Finally, we prove $\tau_{\max}$ can be extended to $\infty$. 
We know that $e_1(t) \in \mathbb{D}', \forall t \in [0, \tau_{\max})$, where $\mathbb{D}'$ is a non-empty compact subset of $\mathbb{D}$.
However, if $\tau_{\max} < \infty$ then according to \cite[Proposition C.3.6]{sontag}, $\exists t' \in [0, \tau_{\max})$ such that $e_1(t) \notin \mathbb{D}$. This leads to a contradiction!
Hence, we conclude that $\tau_{\max}$ can be extended to $\infty$.\\

\noindent\textbf{Stage $k$ ($k\in [2;N]$):}

Differentiating the normalized error $e_k(x_{k},t)$ with respect to time and substituting the corresponding equations of the system dynamics \eqref{eqn:sysdyn}, we get
\begin{align*}
    \dot{e_k} = (\gamma_{k,d}(t))^{-1} ( f_k(x_{k}) + g_k(x_{k})x_{k+1} + w_k  \\ - (\dot{r}_k(z_k,t) + \dot{\gamma}_{k,d}(t) e_k)) : =h_k(t,e_k).
\end{align*}
We also define the constraints for $e_k$ through the open and bounded set $\mathbb{D}:=(-1, 1)^n$. \\

\textit{\underline{Step $(i)$}}:
Since the initial state $x_{k}(0)$ satisfies $-\gamma_{k,i}(0) < x_{k,i}(0) < \gamma_{k,i}(0) ,\forall i \in [1;n]$, the initial normalized error $e_k(0)$ is also within the constrained region $\mathbb{D}$. Further, the tube functions are bounded and continuously differentiable functions of time, the functions $f_k(x_{k})$ and $g_k(x_{k})$ are locally Lipschitz and the control law $r_{k+1}(z_{k+1},t)$ is smooth over $\mathbb{D}$. As a consequence, $h_k(t,e_k)$ is bounded and continuously differentiable on $t$ and locally Lipschitz on $e_k$ over $\mathbb{D}$. 

Therefore, according to \cite[Theorem 54]{sontag}, there exists a maximal solution to the initial value problem $\dot{e_k} = h_k(t,e_k), e_k(0) \in \mathbb{D}$ on the time interval $[0, \tau_{\max})$ such that $e_k(t) \in \mathbb{D}, \forall t \in [0, \tau_{\max}).$\\

\textit{\underline{Step $(ii)$}}:
Consider the following positive definite and radially unbounded Lyapunov function candidate: $V_k = \frac{1}{2}\varepsilon_k^\top\varepsilon_k$. 

Differentiating $V_k$ with respect to $t$ and substituting $\dot{\varepsilon_k}$, $\dot{e_k}$, system dynamics \eqref{eqn:sysdyn}, and the control strategy (\ref{eqn:Control_ras}), we obtain:
\begin{align*}
    &\dot{V}_k = \varepsilon_k^T \dot{\varepsilon}_k= \varepsilon_k^T \frac{2}{1-e_k^Te_k}\dot{e}_k \\
    & = \frac{1}{2} \varepsilon_k^T \xi_k \left(\dot{x}_k - \dot{r}_k (z_k,t) - \dot{\gamma}_{k,d}e_k\right) \\
    &= \frac{1}{2} \varepsilon_k^T \xi_k \big( f_k(x_{k}) + g_k(x_{k})x_{k+1} +w_k  - \dot{r}_k(z_k,t) - \dot{\gamma}_{k,d}e_k\big) \\
    &= \frac{1}{2} \varepsilon_k^T \xi_k \big( f_k(x_{k}) -\kappa_k g_k(x_{k}) \xi_k \varepsilon_k + w_k - \dot{r}_k(z_k,t) - \dot{\gamma}_{k,d}e_k\big).
\end{align*}
Using Rayleigh-Ritz inequality and Assumption \ref{assum:pd},
\begin{align*}
    \underline{g_k}\|\varepsilon_k\|^2 \|\xi_k\|^2 & \leq \lambda_{\min}(g_k(x_{k}))\|\varepsilon_k\|^2 \|\xi_k\|^2 \\ & \leq \varepsilon_k^\top \xi_k g_k(x_{k}) \xi_k \varepsilon_k, \\
    \implies -\frac{1}{2}\kappa_k \varepsilon_k^\top \xi_k g_k(x_{k}) \xi_k \varepsilon_k & \leq -\frac{1}{2} \kappa_k \underline{g_k}\|\varepsilon_k\|^2 \|\xi_k\|^2 \\ &= - \kappa_g^k\|\varepsilon_k\|^2 \|\xi_k\|^2.
\end{align*}
Therefore, 
$\dot{V}_k \leq -\kappa_g^k\|\varepsilon_k\|^2 \|\xi_k\|^2 + \|\varepsilon_k\| \|\xi_k\| \|\Phi_k\|,$
where $\Phi_k := \frac{1}{2}\left(f_k(x_{k}) + w_k - \dot{r}_k(z_k,t) - \dot{\gamma}_{k,d}e_k\right)$. 
From the construction of $ \gamma_{k,d}$ we know that $\dot{\gamma}_{k,d}$ is bounded. From step k-a, we have $e_k(t) \in \mathbb{D}$ and consequently $\forall i \in [1;n]$, $x_{k,i}(t) \in (-\gamma_{k,i}(t), \gamma_{k,i}(t))$. Thus, owing to the continuity of $f_k(x_{k})$ and employing the extreme value theorem, we can infer $\|f_k(x_{k})\| < \infty$. Also, since $r_k$ is bounded, we can say that $\dot{r}_k$ is bounded.  
Hence, $\|\Phi_k\| < \infty, \forall t \in [0, \tau_{max})$. 

Now add and substract $\kappa_g^k\theta\norm{\varepsilon_k}^2 \|\xi_k\|^2$, where $\theta\in(0,1)$.
\begin{align*}
    \dot{V_k} &\leq -\kappa_g^k(1-\theta)\norm{\varepsilon_k}^2 \|\xi_k\|^2 - \\ &\quad \quad \quad \norm{\varepsilon_k}\|\xi_k\| \left(\kappa_g^k \theta \norm{\varepsilon_k}\|\xi_k\| - \|\Phi_k\| \right) \\
    &\leq -\kappa_g^k(1-\theta)\norm{\varepsilon_k}^2 \norm{\xi_k}^2, \forall \kappa_g^k \theta \norm{\varepsilon_k}\|\xi_k\| - \|\Phi_k\| \geq 0 \\
    &\leq -\kappa_g^k(1-\theta)\norm{\varepsilon_k}^2 \|\xi_k\|^2, \forall \norm{\varepsilon_k} \geq \frac{\|\Phi_k\|}{\kappa_g^k \theta \|\xi_k\|} \nonumber,\\
    & \forall t \in [0, \tau_{\max}).
\end{align*}
Therefore, we can conclude that there exists a time-independent upper bound $\varepsilon_k^* \in \mathbb{R}_{0}^+$ to the transformed error $\varepsilon_k$, i.e., $\|\varepsilon_k\| \leq \varepsilon_k^*, \forall t \in [0, \tau_{\max})$.

Consequently, taking inverse logarithmic function,
\begin{align*}
    -1 < \frac{e_{k,i}^{-\varepsilon_{k,i}^*}-1}{e_{k,i}^{-\varepsilon_{k,i}^*}+1} =: e_{k,i,L} \leq e_{k,i} \leq e_{k,i,U} := \frac{e_{k,i}^{\varepsilon_{k,i}^*}-1}{e_{k,i}^{\varepsilon_{k,i}^*}+1} < 1, \nonumber \\
    \forall t \in [0, \tau_{\max}), \ \text{for } i \in [1;n].
\end{align*}

Therefore, by employing the control law (\ref{eqn:Control_ras}), we can constrain $e_k$ to a compact subset of $\mathbb{D}$ as:
\begin{align*}
    e_k(t) \in [e_{k,L}, e_{k,U}] =: \mathbb{D}' \subset \mathbb{D}, \forall t \in [0, \tau_{\max}), \nonumber
\end{align*}
$e_{k,L} = [e_{k,1,L}, \ldots, e_{k,n,L}]^\top$, $e_{k,U} = [e_{k,1,U}, \ldots, e_{k,n,U}]^\top$.\\

\textit{\underline{Step $(iii)$}}:
Finally, we prove that $\tau_{\max}$ can be extended to $\infty$. 
We know that $e_k(t) \in \mathbb{D}', \forall t \in [0, \tau_{\max})$, where $\mathbb{D}'$ is a non-empty compact subset of $\mathbb{D}$.
However, if $\tau_{\max} < \infty$ then according to \cite[Proposition C.3.6]{sontag}, $\exists t' \in [0, \tau_{\max})$ such that $e_k(t) \notin \mathbb{D}$. This leads to a contradiction.
So, we conclude that $\tau_{\max}$ can be extended to $\infty$.


Thus, the control strategy in \eqref{eqn:Control_ras} ensures the agent's trajectory remains within their STTs, thereby completing the proof.
\end{proof}

\begin{theorem}\label{th:T-RAS_CA}
    Given the control input for $j$-th agent, for all $j \in [1;M] $, as discussed in Theorem \ref{theorem_inside}, the agents in the multi-agent systems will satisfy their local temporal reach-avoid-stay (T-RAS) tasks while avoiding inter-agent collisions.
\end{theorem}
\begin{proof}
    As derived in Section III, the STTs guarantee to reach the corresponding target within the specified time. Since the control input from Theorem \ref{theorem_inside} guarantees that the system trajectory of $j$-th agent is inside its corresponding tubes, for all $j \in [1;M]$, it ensures that the T-RAS specification is fulfilled for each agent.\\
    Also, the tubes are collision-free by design due to Equation \eqref{eq:collision}. Hence, keeping the trajectory of each agent inside the tube will formally guarantee the avoidance of inter-agent collision, thereby completing the proof.
\end{proof}

\begin{remark}
Note that the closed-form time-dependent control law \eqref{eqn:Control_ras} is decentralized and approximation-free, since it does not depend on the knowledge of the system dynamics. The gain $\kappa_i$ is a user-defined positive design parameter that regulates how aggressively the control law ensures the system state remains within the STT. Additionally, if $g_{i,s}^{(j)}$ is negative definite, $\kappa_{i}^{(j)}$ (in control law \eqref{eqn:Control_ras}) $\in \R \setminus \R_0^+$.
\end{remark}

\section{Case Study}\label{case-study}
Here, we consider two different cases to show the efficacy of our proposed approach. First, we consider multiple omnidirectional robots, and multiple drones. All computations were performed on a machine with an Intel i7-7700 CPU, 32GB RAM, and a Linux Ubuntu operating system. 

\subsection{Omnidirectional Robots}

The states of the $j$-th omnidirectional robot, which is adapted from \cite{NAHS}, are defined as 
\begin{align}
    \begin{bmatrix}
        \dot{x}_1^{(j)} \\ \dot{x}_2^{(j)} \\ \dot{x}_3^{(j)}
    \end{bmatrix}
    = 
    \begin{bmatrix}
        \cos{x_3^{(j)}} & -\sin{x_3^{(j)}} & 0 \\ \sin{x_3^{(j)}} & \cos{x_3^{(j)}} & 0 \\ 0 & 0 & 1
    \end{bmatrix}
    \begin{bmatrix}
        v_1^{(j)} \\ v_2^{(j)} \\ \omega^{(j)}
    \end{bmatrix} + w(t),
\end{align}
where the state vector $[x_1^{(j)}, x_2^{(j)}, x_3^{(j)}]^\top$ captures the pose of the $j$-th robot, $[v_1^{(j)}, v_2^{(j)}, \omega^{(j)}]^\top$ is the input velocity vector in the frame of $j$-th robot, and $w$ is an external disturbance. 

We consider a group of four robots, moving in a $5 \times 5$ arena with the unsafe set positioned as shown in Figure \ref{fig:sim1}. The start and goal locations for the robots are given in Table \ref{tab:start_goal_omni}.
\begin{table}[ht]
    \centering
    \caption{Start and goal location of each robot}
    \begin{tabular}{ccc}
    \hline
     Robot no. & Start location & Goal location  \\ \hline
     1 & $[4.5,5]\times[0,0.5]$ & $[0,0.5]\times[0,0.5]$ \\
     2 & $[2.5,3]\times[0,0.5]$ & $[4.5,5]\times[0,0.5]$ \\
     3 & $[4.5,5]\times[4.5,5]$ & $[2.5,3]\times[0,0.5]$ \\
     4 & $[0,0.5]\times[0,0.5]$ & $[4.5,5]\times[4.5,5]$ \\
     \hline
    \end{tabular}
    \label{tab:start_goal_omni}
\end{table}

The time bound for reaching the target is $t_c = 10$. The unsafe set (red) is assumed to be present at the same location for the entire time-horizon $[0,t_c]$. As discussed in Section \ref{data_driven}, we consider the template for the STTs as second-order polynomials in time $t$ for robots 1-3 and third-order polynomials in time $t$ for robot 4. The STTs for all robots obtained to solve the T-RAS and CA task for this case are given by the curves in Table \ref{tab:stt1}.


\begin{table}[ht]
\centering
\caption{STTs of each robot}
\resizebox{0.5\textwidth}{!}{
\begin{tabular}{ccll}
\hline
\multicolumn{1}{l}{Robot} & \multicolumn{1}{l}{Dim} & $\gamma_L(c_L,t)$                            & $\gamma_U(c_U,t)$                              \\ \hline
\multirow{2}{*}{1}        & 1                       & $4.5 - 0.8955t + 0.0445t^2$           & $ 5 - 1.0515t + 0.0602t^2$              \\
                          & 2                       & $0 + 0t + 0t^2$                       & $0.5 - 0.156t + 0.0156t^2$              \\ \hline
\multirow{2}{*}{2}        & 1                       & $0 + 1.0347t - 0.0585t^2$             & $0.5 + 0.8787t - 0.0429t^2$             \\
                          & 2                       & $2.5 + 0.4264t - 0.0676t^2$           & $3 + 0.2704t - 0.052t^2$                \\ \hline
\multirow{2}{*}{3}        & 1                       & $4.5 + 0.0883t - 0.0538t^2$           & $5 - 0.0678t -0.0382t^2$                \\
                          & 2                       & $4.5 + 0.1665t - 0.0367t^2$           & $5 + 0.0105t - 0.0211t^2$               \\ \hline
\multirow{2}{*}{4}        & 1                       & $0 + 3.9463t - 0.9857t^2 + 0.0636t^3$ & $0.5 + 3.8711t - 0.9928t^2 + 0.0651t^3$ \\
                          & 2                       & $0 + 0.4283t + 0.0009t^2 + 0.0001t^3$ & $0.5 + 0.1945t + 0.0422t^2 - 0.0017t^3$ \\ \hline
\end{tabular}
}
\label{tab:stt1}
\end{table}

The Lipschitz constants $\mathcal{L}_L$ and $\mathcal{L}_U$ are estimated to be $3.946$ and $3.871$ using Algorithm \ref{algo:Lipschitz}, so $\mathcal{L} = 7.817$ and consider $\epsilon = 0.002$. We have used the Z3 SMT solver \cite{z3} to solve the Optimization problem. The $\eta_S^*$ upon solving the SOP is $-0.05$, which essentially means $\eta_S^* + \mathcal{L}\epsilon = -0.05 + 7.817\times0.002 = -0.0344< 0$, that follows Theorem \ref{th:constr}. The computation time to obtain the STTs via solving SOP in \eqref{eq:SOP} is 208 seconds.

The trajectory of the robots under the influence of the proposed control law in \eqref{eqn:Control_ras} with the STTs for each robot in each dimension, are shown in Figure \ref{fig:sim1}. 

\begin{figure*}[t]
    \centering
    \includegraphics[width =0.85\textwidth]{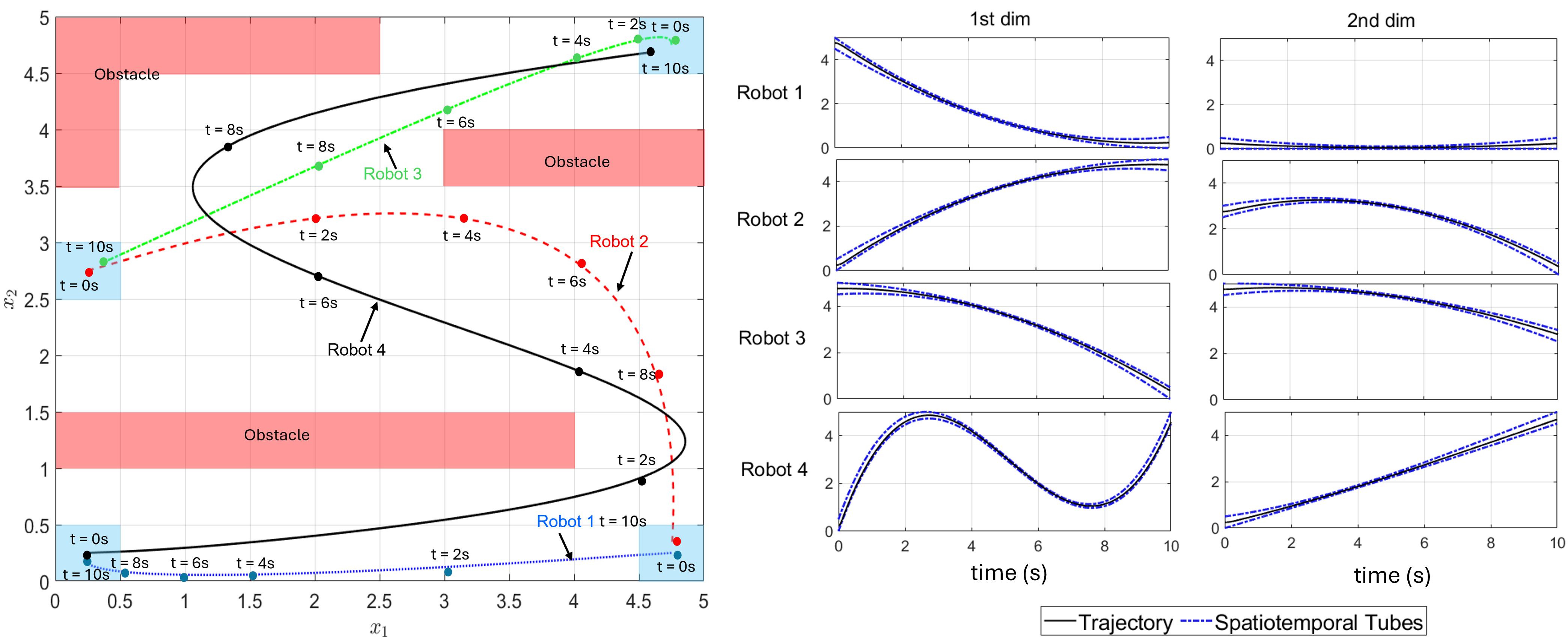}
    \caption{Spatiotemporal tubes for T-RAS tasks in omnidirectional robots.}
    \label{fig:sim1}
\end{figure*}

\begin{figure*}[t]
    \centering
    \includegraphics[width =0.85\textwidth]{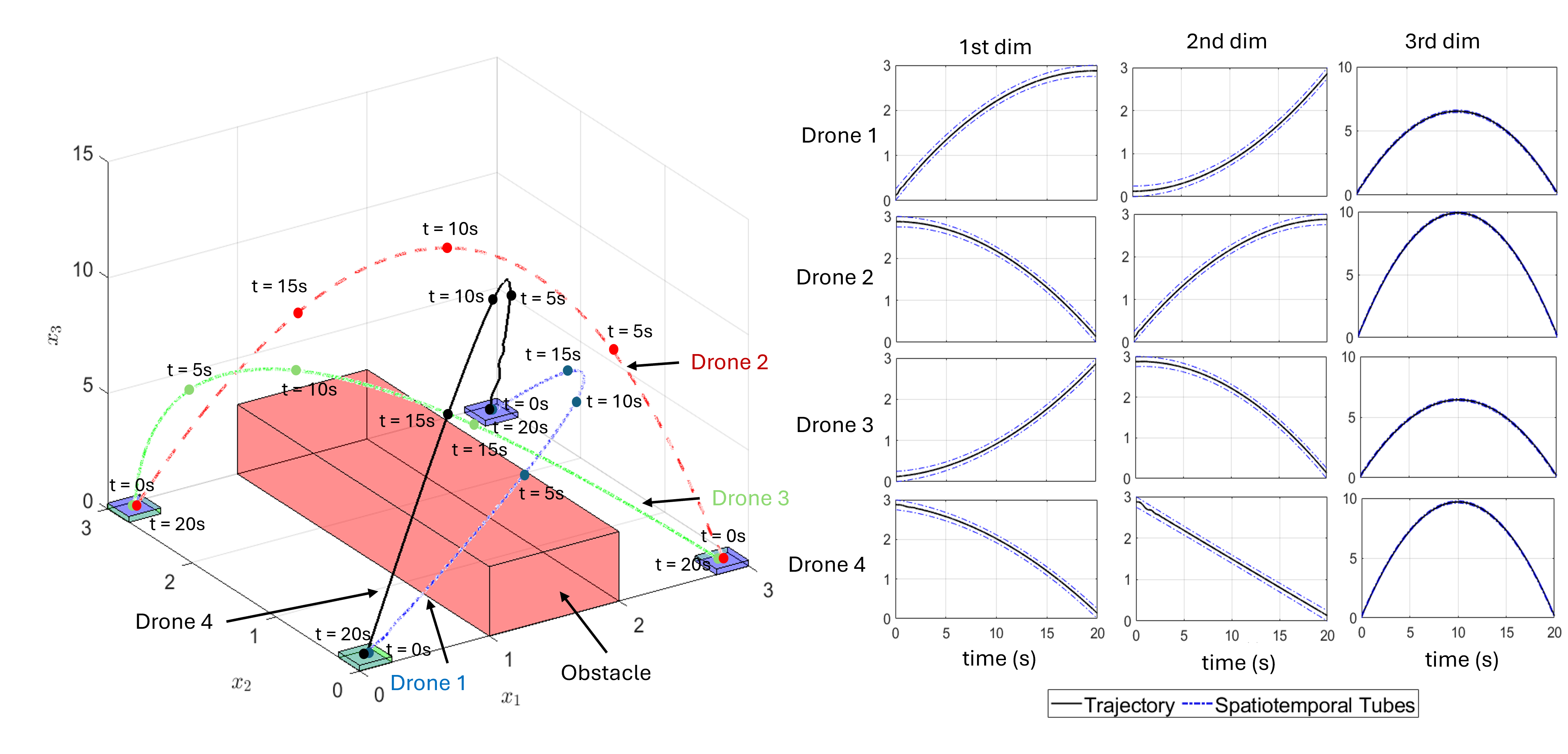}
    \caption{Spatiotemporal tubes for T-RAS tasks in drones.}
    \label{fig:sim2}
\end{figure*}

\subsection{Euler-Lagrange System- Multiple Drones}
The dynamics of the $j$-th drone is adapted from \cite{APF}, 
\begin{align}
    \left[\dot{x}^{(j)} \ \dot{y}^{(j)} \ \dot{z}^{(j)} \right]^\top    &= \left[v_x^{(j)} \ v_y^{(j)} \ v_z^{(j)} \right]^\top + w_1(t), \\
    \left[v_x^{(j)} \ v_y^{(j)} \ v_z^{(j)} \right]^\top &= \left[u_x^{(j)} \ u_y^{(j)} \ u_z^{(j)} \right]^\top + w_2(t),
\end{align}
where $[x^{(j)},y^{(j)},z^{(j)}]^\top$ captures the position of the $j$-th drone, $[v_x^{(j)},v_y^{(j)},v_z^{(j)}]^\top$ is the velocity of the $j$-th drone and $[u_x^{(j)},u_y^{(j)},u_z^{(j)}]^\top$ is the control input of the $j$-th drone. 

We consider a group of 4 drones, diagonally moving from one corner to another corner of a $[0,3]\times[0,3]\times[0,15]$ arena. The unsafe set is given by, $\U = [1,2]\times[0,3]\times[0,3]$. The start and goal location of each drone is given in Table \ref{tab:start_goal_drone}. 

\begin{table}[ht]
    \centering
    \caption{Start and goal location of each drone}
    \resizebox{0.5\textwidth}{!}{
    \begin{tabular}{ccc}
    \hline
     Robot no. & Start location & Goal location  \\ \hline
     1 & $[0,0.25]\times[0,0.25]\times[0,0.25]$ & $[2.75,3]\times[2.75,3]\times[0,0.25]$ \\
     2 & $[2.75,3]\times[0,0.25]\times[0,0.25]$ & $[0,0.25]\times[2.75,3]\times[0,0.25]$ \\
     3 & $[0,0.25]\times[2.75,3]\times[0,0.25]$ & $[2.75,3]\times[0,0.25]\times[0,0.25]$ \\
     4 & $[2.75,3]\times[2.75,3]\times[0,0.25]$ & $[0,0.25]\times[0,0.25]\times[0,0.25]$  \\
     \hline
    \end{tabular}
    }
    \label{tab:start_goal_drone}
\end{table}

The time bound for reaching the target is considered as $t_c = 20$. The unsafe set (red colored) is assumed to present at the same location for the complete time-horizon $[0,t_c]$. As discussed in Section \ref{data_driven}, we consider the template for the STTs as second-order polynomials in time $t$. The STTs for all the drones obtained to solve the T-RAS and CA task for this case are given by the curves in Table \ref{tab:stt2}.


\begin{table}[ht]
\centering
\caption{STTs of each drone}
\resizebox{0.5\textwidth}{!}{
\begin{tabular}{cccc}
\hline
Drone           & Dim & $\gamma_L(c_L,t)$            & $\gamma_L(c_L,t)$             \\ \hline
\multirow{3}{*}{1} & 1   & $0 + 0.2845t - 0.0074t^2$    & $ 0.25 + 0.2745t - 0.0069t^2$ \\
                   & 2   & $0 + 0.0068t + 0.0065t^2$    & $0.25 - 0.0032t + 0.007t^2$   \\
                   & 3   & $0 + 1.2874t - 0.0644t^2$    & $0.25 + 1.2774t - 0.0639t^2$  \\ \hline
\multirow{3}{*}{2} & 1   & $2.75 + 0.0026t - 0.007t^2$  & $3 - 0.0074t - 0.0065t^2$     \\
                   & 2   & $0 + 0.2767t - 0.007t^2$     & $0.25 + 0.2667t - 0.0065t^2$  \\
                   & 3   & $0 + 1.96t - 0.098t^2$       & $0.25 + 1.95t - 0.0975t^2$    \\ \hline
\multirow{3}{*}{3} & 1   & $0 + 0.0258t + 0.0056t^2$    & $0.25 + 0.0156t + 0.0061t^2$  \\
                   & 2   & $2.75 + 0.0117t - 0.0075t^2$ & $3 + 0.0017t - 0.007t^2$      \\
                   & 3   & $0 + 1.2658t - 0.0633t^2$    & $0.25 + 1.2558t - 0.0628t^2$  \\ \hline
\multirow{3}{*}{4} & 1   & $2.75 - 0.0296t - 0.0054t^2$ & $3 - 0.0396t - 0.0049t^2$     \\
                   & 2   & $2.75 - 0.1336t - 0.0002t^2$ & $3 - 0.1436t + 0.0003t^2$     \\
                   & 3   & $0 + 1.9175t - 0.0959^2$     & $0.25 + 1.9075t - 0.0954t^2$  \\ \hline
\end{tabular}
}
\label{tab:stt2}
\end{table}

The trajectory of the drones under the influence of the proposed control law in \eqref{eqn:Control_ras} with the STTs for each robot in each dimension, are shown in Figure \ref{fig:sim2}. 

\underline{\textit{Comparison:}}
Path planning algorithms \cite{qin2023review} offer a potential solution to the case studies, albeit lacking formal guarantees of solution. Even though Model Predictive Control \cite{dai2017distributed} offers a better solution to the problem, it lacks the decentralized control policy and requires approximate knowledge of the model dynamics. Conversely, symbolic control techniques \cite{tabuda, SCOTS} promise formal guarantees but come at the expense of increased computational complexity. Reinforcement learning-based approaches \cite{liu2020mapper}, \cite{gupta2017cooperative}, though they can handle the unknown dynamics, fail to provide a formal guarantee for the given tasks. CBF \cite{C3BF}\cite{Meng3} and Funnel-based control \cite{lindemann2019feedback} can provide formal guarantee to the task, but fail in case of the unknown dynamical system. Compared to these approaches, the STT framework provides a formal guarantee to solve the T-RAS problem and avoid inter-agent collision under a decentralized control policy. The STT, moreover, can handle the specifications for unknown dynamical systems. The comparison of different aspects of spatiotemporal tubes with state-of-the-art algorithms is shown in Table \ref{tab:comp}.

\begin{table*}[t]
\centering
\begin{threeparttable}
\caption{Comparing Spatiotemporal tubes with classical algorithms}
\begin{tabular}{lcccccccccccc}
\hline
\textbf{Algorithm} & \multicolumn{2}{c}{\textbf{\begin{tabular}[c]{@{}l@{}}Closed-form \\ Control \end{tabular}}} & \multicolumn{2}{c}{\textbf{\begin{tabular}[c]{@{}l@{}}Formal \\ Guarantee \end{tabular}}} & \multicolumn{2}{c}{\textbf{\begin{tabular}[c]{@{}l@{}}Prescribed-time \\ Reachability \end{tabular}}} & \multicolumn{2}{c}{\textbf{\begin{tabular}[c]{@{}l@{}}Unknown \\ Dynamics \end{tabular}}} & \multicolumn{2}{c}{\textbf{\begin{tabular}[c]{@{}l@{}}Bounded \\ Disturbance \end{tabular}}} & \multicolumn{2}{c}{\textbf{\begin{tabular}[c]{@{}l@{}}Time dependent \\ Obstacle \end{tabular}}} \\ \hline
RRT*\cite{M-RRT} & \multicolumn{2}{c}{-\tnote{1}} & \multicolumn{2}{c}{\xmark} & \multicolumn{2}{c}{-\tnote{2}} & \multicolumn{2}{c}{\cmark} & \multicolumn{2}{c}{-\tnote{3}} & \multicolumn{2}{c}{\cmark} \\
MPC\cite{MPC1}\cite{dai2017distributed} & \multicolumn{2}{c}{\xmark} & \multicolumn{2}{c}{\xmark} & \multicolumn{2}{c}{\xmark} & \multicolumn{2}{c}{\xmark} & \multicolumn{2}{c}{\xmark} & \multicolumn{2}{c}{\cmark} \\ 
RL\cite{liu2020mapper,gupta2017cooperative} & \multicolumn{2}{c}{\xmark} & \multicolumn{2}{c}{\xmark} & \multicolumn{2}{c}{\xmark} & \multicolumn{2}{c}{\cmark} & \multicolumn{2}{c}{\cmark} & \multicolumn{2}{c}{\cmark} \\
CBF-based methods \cite{C3BF}\cite{Meng3} & \multicolumn{2}{c}{\xmark} & \multicolumn{2}{c}{\cmark} & \multicolumn{2}{c}{\xmark} & \multicolumn{2}{c}{\xmark} & \multicolumn{2}{c}{\xmark} & \multicolumn{2}{c}{\xmark} \\
Symbolic Control\cite{SCOTS} & \multicolumn{2}{c}{\xmark} & \multicolumn{2}{c}{\cmark} & \multicolumn{2}{c}{\xmark} & \multicolumn{2}{c}{\xmark} & \multicolumn{2}{c}{\xmark} & \multicolumn{2}{c}{\xmark} \\
Funnel-based Control\cite{lindemann2019feedback} & \multicolumn{2}{c}{\cmark} & \multicolumn{2}{c}{\cmark} & \multicolumn{2}{c}{\xmark} & \multicolumn{2}{c}{\cmark} & \multicolumn{2}{c}{\cmark} & \multicolumn{2}{c}{\xmark} \\
STT (proposed) & \multicolumn{2}{c}{\cmark} & \multicolumn{2}{c}{\cmark} & \multicolumn{2}{c}{\cmark} & \multicolumn{2}{c}{\cmark} & \multicolumn{2}{c}{\cmark} & \multicolumn{2}{c}{\cmark} \\
\hline
\end{tabular}

\label{tab:comp}
\begin{tablenotes}
    \item [1] Additional mechanisms like PID and MPC are required for control.
    \item [2] Additional mechanisms are required to satisfy reachability within the prescribed time.
    \item [3] Additional mechanisms are required to handle bounded disturbances.
\end{tablenotes}
\end{threeparttable}
\end{table*}

\section{Conclusion and Future work}\label{conclusion}
The paper focuses on constructing the spatiotemporal tubes using a sampling-based approach from the available data of temporal unsafe sets, aiming to ensure the T-RAS task and prevent the collision between the tubes. Keeping the agents' trajectories within their corresponding tubes formally guarantees that the agents do not collide with each other while satisfying their respective T-RAS tasks. We formulate an SOP to gather obstacle data and model the tube so that it ensures that the T-RAS task is completed with a certified confidence of 1 and that no tubes collide with each other. Consequently, we can design a closed-form, approximation-free control law to retain the unknown systems' trajectories within the tubes. We showcase the success of our approach through two different case studies. 

In our future work, we plan to expand our approach by developing solutions that can accommodate arbitrary input constraints. This will allow for greater flexibility and applicability of our methods across different scenarios and systems, making them more robust and versatile.
Additionally, we aim to explore the application of STTs to satisfy STL specifications, thereby broadening the potential use cases and ensuring that our methods can be applied to a diverse set of tasks in future research.

\bibliographystyle{unsrt} 
\bibliography{sources} 

\appendix

\textbf{Computation of Lipschitz constants $\mathcal{L}_L$ and $\mathcal{L}_U$}
\begin{algorithm}
    \caption{Estimation of $\mathcal{L}_L$ using data}
    \label{algo:Lipschitz}
    \begin{algorithmic}[1]
        \State Select two time instances randomly $t_k$, $t_m$ such that $\lVert t_k - t_m\rVert \leq \alpha, \forall k,m \in [1;\overline{N}], \alpha \in \R_{>0}$ 
        \State Calculate $\theta_{i,k}^{(j)} = \gamma_{i,L}(c_{i,L}^{(j)},t_k) $ and $ \theta_{i,m}^{(j)} = \gamma_{i,L}(c_{i,L}^{(j)},t_m)$ where $\theta_{i,k}^{(j)}$ and $\theta_{i,m}^{(j)}$ denotes the lower bound of the tube in $i$-th dimension at $k$-th and $m$-th time instances for $j$-th agent. 
        \State Compute slope $S_{i,km}^{(j)} = \frac{\lVert \theta_{i,k}^{(j)} - \theta_{i,m}^{(j)} \rVert}{t_k - t_m} , \quad \forall j \neq k$ 
        \State Compute maximum slope as $\Psi_i^{(j)} = \max \{S_{i,km}^{(j)}|k,m \in [1;\overline{N}], k \neq m\}$ 
        \State Repeat steps 1-4 $R$ times and obtain $\Psi_{i,1}^{(j)}, \ldots, \Psi_{i,R}^{(j)}$ 
        \State Apply Reverse Weibull Distribution \cite{Lipschitz2} to $\Psi_{i,1}^{(j)}, \ldots, \Psi_{i,M}^{(j)}$, which gives us so-called location, scale and shape parameters
        \State The obtained \textit{location parameter} is denoted by $\mathcal{L}_i^{(j)}$
        \State Repeat steps 1-7 to get $\mathcal{L}_1^{(j)}, \ldots, \mathcal{L}_n^{(j)}$ 
        \State $\mathcal{L}_L^{(j)} = \max \{\mathcal{L}_1^{(j)}, \ldots, \mathcal{L}_n^{(j)}\}$
        \State Repeat steps 1-9 to get $\mathcal{L}_L^1, \ldots, \mathcal{L}_L^M$
        \State $\mathcal{L}_L = \max \{\mathcal{L}_L^1, \ldots, \mathcal{L}_L^M\}$
    \end{algorithmic}
\end{algorithm}

Employing the results of \cite{Lipschitz2}, we propose the Algorithm \ref{algo:Lipschitz} to estimate the Lipschitz constants of the STTs using a finite number of data collected from the tube boundary. Though we introduce the algorithm for the computation of $\mathcal{L}_L$, one can leverage a similar algorithm to estimate $\mathcal{L}_U$, following the same procedure.

 The following Lemma \ref{lem:Lipschitz}, borrowed from \cite{Lipschitz2}, under the proposed algorithm ensures the convergence of the estimated Lipschitz constants to their actual values.

\begin{lemma}[\cite{FV_DD}]\label{lem:Lipschitz}
    In the verge of Algorithm \ref{algo:Lipschitz}, the estimated Lipschitz constants $\mathcal{L}_L$ and $\mathcal{L}_U$ tends to their actual values if and only if $\alpha$ goes to zero with $\overline{N}, M$ tends to infinity.
\end{lemma}
\quad Note that, picking minimal value of $\alpha$ and very high value of $\overline{N}, M$ will give a precise approximation of the Lipschitz constants.

\end{document}